\documentclass[journal,twocolumn]{paper}  

\usepackage[T1]{fontenc}
\usepackage[utf8]{inputenc}
\usepackage{mathtools, bm}
\usepackage{amssymb, bm}
\usepackage{multirow}
\usepackage{wrapfig}
\usepackage{nicefrac}
\usepackage{lscape}
\usepackage{styleJK_journal}
\usepackage{caption}
\usepackage{subcaption}
\usepackage{pifont}

\usepackage{amsthm}

\theoremstyle{plain}

\newenvironment{rem}[1][Remark]{\begin{trivlist}
\item[\hskip \labelsep \textit{Remark} (#1).]}{\hfill$\lrcorner$\end{trivlist}}

\newcounter{ctprop}
\newenvironment{prop}[1][Proposition]{\addtocounter{ctprop}{1}\begin{trivlist}
\item[\hskip \labelsep \textbf{Proposition \arabic{ctprop}}.]}{\end{trivlist}}

\newcounter{cttheorem}
\newenvironment{theorem}[1][Theorem]{\addtocounter{cttheorem}{1}\begin{trivlist}
\item[\hskip \labelsep \textbf{Theorem \arabic{cttheorem}}.]}{\end{trivlist}}

\newcounter{ctProb}
\newenvironment{pblm}[1][Problem]{\addtocounter{ctProb}{1}%
\begin{trivlist}
\item[\hskip \labelsep \textit{Problem} \arabic{ctProb}.]}{\end{trivlist}}

\newcounter{ctlemma}
\newenvironment{lem}[1][Lemma]{\addtocounter{ctlemma}{1}\begin{trivlist}
\item[\hskip \labelsep \textbf{Lemma \arabic{ctlemma}}.]}{\end{trivlist}}

\newcounter{cthyp}
\newenvironment{asm}[1][Assumption]{\addtocounter{cthyp}{1}\begin{trivlist}
\item[\hskip \labelsep \textbf{Assumption \arabic{cthyp}}.]}{\end{trivlist}}

\newcommand\Ca{\mathcal{C}}

\title{\LARGE \bf
Hamiltonian Point of View on Parallel Interconnection of Buck Converters	\\
{\color{blue} Extended version}}

\author{J\'er\'emie Kreiss$^{1}$, Jean-Fran\c cois Tr\'egou\"et$^{1}$, Damien Eberard$^{1}$, \\Romain Delpoux$^{1}$, Jean-Yves Gauthier$^{1}$ and Xuefang Lin-Shi$^{1}$
\thanks{$^{1}$Every author is with Laboratoire Amp\`ere, INSA Lyon,
        Universit\'e de Lyon, 20, Avenue Albert Einstein, 69100 Villeurbanne, France
        {\tt\small firstname.lastname@insa-lyon.fr}.
                Corresponding author mail:
                {\tt\small jeremie.kreiss@insa-lyon.fr }}%
}
\begin{document}

\maketitle
\selectlanguage{English}

{\color{blue}
This report is an extended version of the corresponding paper published in TCST which does not included blue parts of this document.
}
\begin{abstract}
	In this paper, parallel interconnection of DC/DC converters is considered. For this topology of converters feeding a  common load, it has been recently shown that dynamics related to voltage regulation can be completely separated from the current distribution without considering frequency separation arguments, which inevitably limits achievable performance. Within the Hamiltonian framework, this paper shows that this separation between current distribution and voltage regulation is linked to the energy conservative quantities: the Casimir functions. Furthermore, a robust control law is given in this framework to get around the fact that the load might be unknown. In this paper, we also ensure that the system converges to the optimal current repartition, without requiring explicit expression of the optimal locus. Finally, resulting control law efficiency is assessed through experimental results.
\end{abstract}


\section{Introduction}

Nowadays, many applications such as low-voltage/high-current power supplies are composed of several power converters connected to a single load. Indeed, this structure benefits from several advantages as a consequence of the free distribution of load current on each converter. Thereby, it is possible to increase reliability, ease of repair, improve thermal management or reduce output ripple by interleaving phase of Pulse Width Modulation (PWM) for example.

The main challenge on this kind of structure is to regulate output voltage and current distribution together which are coupled dynamics. To cope with this difficulty, most of existing solutions (see e.g. \cite{Thottuvelil1997,Sun2005,Huang2007}) propose control design procedure based on frequency consideration that separates dynamics of the system (output voltage from current distribution). Inevitably, those considerations reduce the achievable performance by imposing slow current distribution dynamics.

However, new solutions have been recently presented \cite{Tregouet2017} where the separation between voltage and current distribution dynamics is geometric. Indeed, by both state and input change of coordinates, those two dynamics are disconnected without any frequency considerations. Hence, this approach provides a framework to easily deal with the two dynamics without sacrificing performance for an arbitrary number of DC/DC buck converters with distinct characteristics.

In this paper, main result of \cite{Tregouet2017} is considered in a different framework: the Port-Controlled-Hamiltonian (PCH) formalism (see \cite{Schaft2014} for more details, \cite{Sira-Ramirez1997,Escobar1999} for power converters in this framework). In addition to describing a large class of non-linear models, the PCH structure intrinsically yields many interesting features such as: (i) energy conservative property, (ii) obvious decomposition between interconnection and damping elements, (iii) straightforward relation that link the dynamics to the energy of the system and (iv) attractive nature of interconnection on ports allowing some Plug\&Play behaviours. 
Furthermore the extension to more complicated converters, potentially non-linear, fits into the PCH frame.

Classical control design methods on PCH models, namely interconnection and damping assignment passivity-based control (IDA-PBC) introduced in \cite{Ortega2002} aim to stabilize the dynamics. Yet, the occurrence of disturbance, uncertainties or reference signal leads to steady-state errors and undesirable behaviour. As we consider that the load is unknown, which is the case in most of practical cases,  we will suffer from this. That is why we resort to robust energy shaping control methods which are developed in the Hamiltonian formalism (see \cite{Romero2013}). In practice, those methods reduce to the addition of an integral action on the PCH ports. Unfortunately, in our case, integral actions on the Hamiltonian ports are not sufficient to deal with uncertainties or disturbances. Yet, interesting developments have been provided in \cite{Ortega2012} and \cite{Donaire2009} where integral action on non-passive outputs is considered.

Main contributions of this paper are as follows: 1) Geometric
decomposition proposed in \cite{Tregouet2017} is completely revisited
within the Hamiltonian framework. A new
change of coordinates is proposed, which can be related to the presence of Casimir function. Here, as compared to
\cite{Tregouet2017}, not only current repartition does not impact
voltage regulation but the opposite also holds: Current distribution is
made independent from voltage regulation. As a result, the system can be
separated into two independent subsystems whereas only cascaded form was
achieved in \cite{Tregouet2017}. 2) Load-independent controller is
proposed and proved to comply with all control specifications, so that
unknown load can be taken into account. Inspired by \cite{Ortega2012}
and \cite{Donaire2009}, where constant exogenous disturbance is
considered, our result is derived relying on integral action on
non-passive outputs via a constructive approach which applies to
\emph{state dependent} disturbance. 3) Optimal current repartition is
achieved at the steady state. Specification about this secondary
objective can be conveniently expressed as an optimization problem.
Using strictly convex cost function satisfying some assumptions, we
are able to design a controller which ensures that the equilibrium point
is the minimum of the cost function even if this minimum is unknown. 4)
As a last contribution, experiment results are provided.

The paper is structured as follows. In Section II, Hamiltonian model of the parallel converters is given as well as the control problem. Section III provides a useful change of coordinates related to Casimir functions in order to separate output voltage dynamics from current repartition; Original control problem is then rewritten in the new coordinates. Control design examples for both known and unknown load 
are described in Section~IV. 
{\color{blue} Discussion about robustness are given in Section~V.} Finally, Section VI presents some experimental results.

{\it Notation:} The notation $x_k$ refers to the $k$-th element of vector $x$, with 1 being the index of first element. Given a function $f:\real^n\rightarrow \real$ we define the operators
\[
	\real^{n\times 1}\ni\nabla_x f\coloneqq \left(\frac{\partial f}{\partial x}\right)^\intercal,\quad \real^{p\times1}\ni\nabla_{\xi}f\coloneqq \left(\frac{\partial f}{\partial {\xi}}\right)^\intercal,
\]
where $\xi\in \real^p$ is a sub-vector of the vector $x$. The symbol $\I_m$ stands for the identity matrix of size $m\times m$. The null matrix of size $m\times n$ is denoted by $\0_{m\times n}$. The vector (column matrix) of size $m$ for which every entry is 1 (respectively 0) is denoted by $\1_m$ (respectively $\0_m$). The operator ``diag'' builds diagonal matrix from entries of the input vector argument.
\section{Problem statement}\label{sec:problem}

In this paper, we are interested in the electrical circuit shown in Fig.~\ref{fig:electrical schematic} which corresponds to parallel interconnection of $m$ heterogeneous and synchronous buck converters sharing a single capacitor $C$ and connected to a common resistive load $R$.
\begin{figure}
	\centering
	\includegraphics[width=\columnwidth]{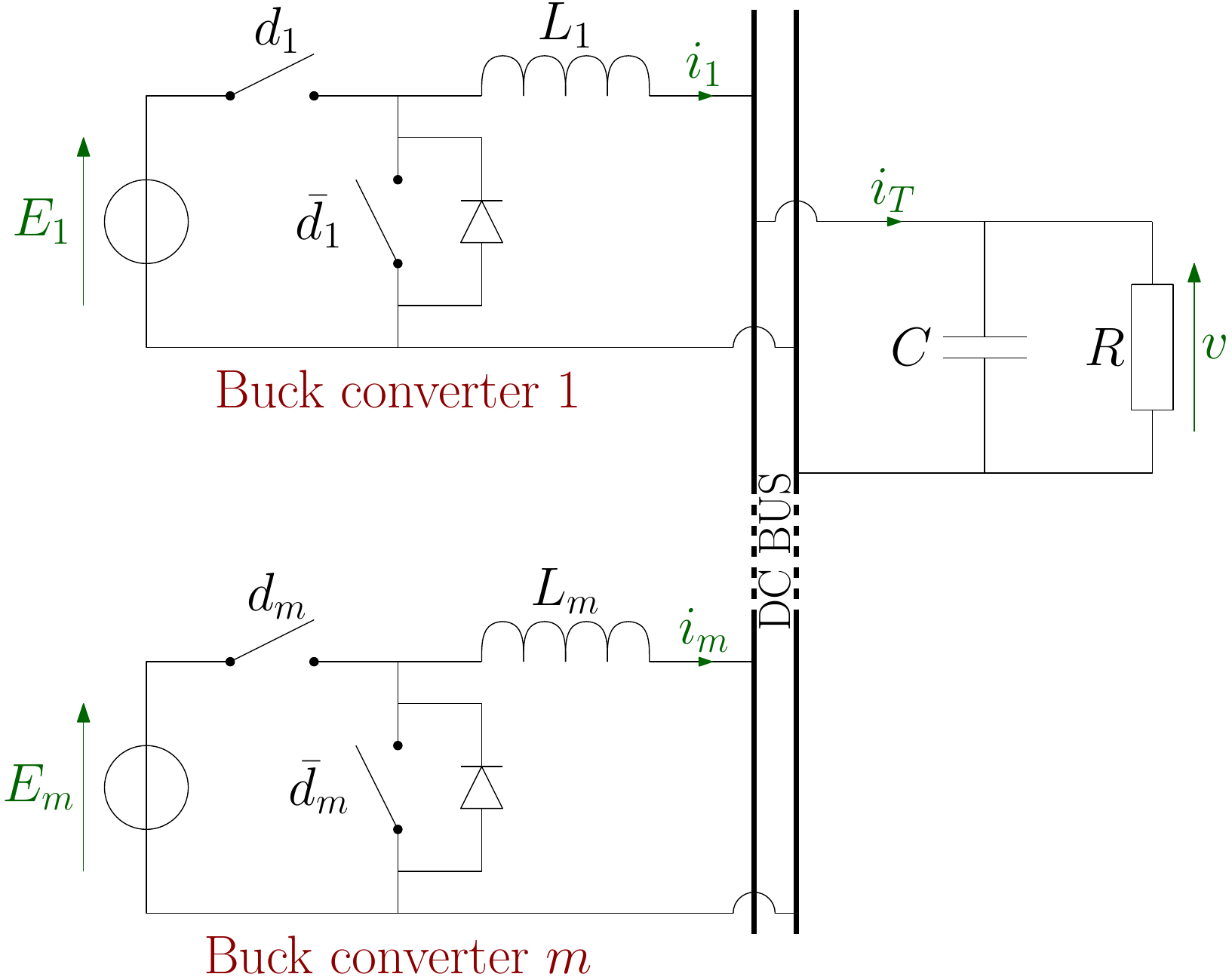}
	\caption{Electrical schematic of $m$ Buck converters}
	\label{fig:electrical schematic}
	\vspace{-0.5cm}
\end{figure}
Instead of acting directly on the switches and dealing with a hybrid system (like e.g.~\cite{Escobar1999}), converters are controlled here via PWM where $d_k$ refers to duty cycle of the $k$-th converter and $\bar d_k = 1-d_k$ is its complementary signal. Index $k$ belongs to the set $\left\{1,\dots,m\right\}$. Capacitor charge is defined as $Q=Cv$ where $v$ is DC bus voltage. Furthermore, we denote by $\varphi_k$ the magnetic flux in $k$-th inductor $L_k$ which is linked to the current $i_k$ by the relation $\varphi_k=L_ki_k$. $E_k$ corresponds to the voltage source of the $k$-th converter.Vector $L\in\real^m$ (resp. $\varphi,i,E\in\real^m$) gathers every element $L_k$ (resp. $\varphi_k,i_k,E_k$) for all $k\in\{1,\dots,m\}$.

 Throughout this paper, we assume that (i) switching frequency $f_s$ is sufficiently large for the dynamics to be approximated by an average continuous time model, and (ii) electrical components and switches are ideal, i.e. parasitic elements can be neglected.

Under those assumptions and using Kirchoff's laws on the energy variables, dynamics of circuit depicted in Fig.~\ref{fig:electrical schematic} are expressed by
\begin{subequations}\label{eq:model}
\begin{eqnarray}
	\forall k\in \left\{1,\dots,m\right\}, ~ \frac{\text{d}\varphi_k}{\text{d}t} =-\frac{Q}{C}+E_k d_k, \label{eq:flux}\\
	\frac{\text{d}Q}{\text{d}t} =\sum_{k=1}^{m}\frac{\varphi_k}{L_k}-\frac{Q}{RC}. \label{eq:charge}
\end{eqnarray}
\end{subequations}
Eq. \eqref{eq:flux} refers to the dynamics of inductors flux produced by each converters whereas \eqref{eq:charge} describes the output capacitor charge dynamics.

This leads to the following linear Hamiltonian system (see \cite{Sira-Ramirez1997})
\begin{equation}\label{eq:Hamiltmodel}
		\dot x=(\mc J -\mc R)\nabla_x H(x)+Bd
\end{equation}
where $\real^{m+1}\ni x=\left[\varphi^\intercal, Q\right]^\intercal$ gathers the energy variables and the smooth function 
\begin{equation}\label{eq:H}
	H(x)=\frac{1}{2}x^\intercal \diag{L,C}^{-1}x
\end{equation}
represents the total stored energy. Energy dissipation is characterized by the $n\times n$ symmetric positive semi-definite matrix
\[
	\mc R=\mc R^\intercal=\diag{\begin{bmatrix} \0_m^\intercal & 1/R	\end{bmatrix}}\geq 0,
\]	
while the skew-symmetric matrix 
\[
	\mc J=-\mc J^\intercal = \begin{bmatrix}\0 & -\1_m\\\1_m^\intercal & 0\end{bmatrix}\in \real^{n\times n},
\]
together with
\begin{equation}
	\label{eq:B}
	B = \begin{bmatrix}\text{diag}\left\{ E\right\} \\
				{\bf 0}_{m}^{\textsf{T}}
			\end{bmatrix}\in \real^{n\times m},
\end{equation}
represent the interconnection structure. See \cite{Schaft2014,Schaft2000,Maschke1992} for more details about PCH modelling.

Bus voltage regulation to a given value $v_r\in \real_{>0}$ (or equivalently $\real_{>0}\ni Q_r\coloneqq Cv_r$) represents the main control objective. We will see later on that the voltage $v$ and hence $Q$ as well only depends on the total current, i.e. the sum of each $i_k$. Thus additional degrees of freedom remain in the way this total current is distributed among the converters. Thereby this paper considers current control distribution or, equivalently, flux control distribution as an additional control objective which is independent from voltage regulation. This secondary control objective is conveniently expressed as a cost function to be minimized, corresponding for example to power losses. This leads to the following expression of the problem addressed in this paper: The constraints correspond to the main control objective whereas the optimization form is related to the secondary objective.

\begin{pblm}
	Given a cost function $J:\real^{m+1} \rightarrow \real$, design state feedback control law $x \mapsto d$ such that resulting closed-loop system admits a (unique) equilibrium point
\[x^\star\coloneqq \begin{bmatrix}\varphi^\star\\ Q^\star\end{bmatrix}\coloneqq\underset{\varphi, Q}{\text{argmin }} J(\varphi, Q)\text{ s.t. }\left\{\begin{array}{l}
	Q=Q_r\\
	\begin{bmatrix}\dot \varphi\\\dot Q\end{bmatrix}=\0
	\end{array}\right.\] which is globally and asymptotically stable (GAS).
\end{pblm}

Literature for Problem 1 is well surveyed in \cite{Luo1999,Huang2007}. Remarkably, almost all existing solutions make use of a two nested loops scheme. The first loop aims associating a close control to each converter, making them acts as a controlled voltage or current source. The second loop is an outer controller whose goal is to achieve exact voltage and current regulation. The fundamental tool for achieving closed-loop stability is frequency separation between those two loops. However, accelerate the outer loop in order to enhance the voltage dynamics might break the frequency separation which, in turn, leads to instability.

In stark contrast with this approach, strategy proposed in the sequel does not required any frequency separation, by resorting to a peculiar change of coordinates.

\section{Change of coordinates}

Purpose of current section is to provide both state and input change of coordinates which aim to separate dynamics into: (i) the minimal part of state vector related to voltage dynamics, i.e. total current and voltage and (ii) the remaining dynamics that is current distribution.

\subsection{Separation of dynamics}

From Fig.~\ref{fig:electrical schematic}, it is clear that the output voltage $v$ only depends on the sum of currents $i_k$ rather than current of each branch individually. 
The following change of coordinates aims to highlight this observation. Accordingly, we introduce the total flux related to the total current as
\begin{equation}
	\varphi_T:=L_{\text{eq},m}\sum_{k=1}^{m}\frac{\varphi_k}{L_k} = \1_m^\intercal \diag{L}^{-1} L_{\mathrm{eq},m} \varphi
\end{equation}
with $L_{\text{eq},k}$ being the equivalent inductor of  the $k$ first coils connected in parallel  (see Fig.~\ref{fig:Leq} for $k=m$):\footnote{$L_{\text{eq},m}$ has been introduced in order to make $\varphi_T$ homogeneous to magnetic flux.}
\begin{equation*}
	\frac{1}{L_{\text{eq},k}}\coloneqq\sum_{j=1}^k\frac{1}{L_j}\neq 0.
\end{equation*}
\begin{figure}
	\centering
	\includegraphics[width=6cm]{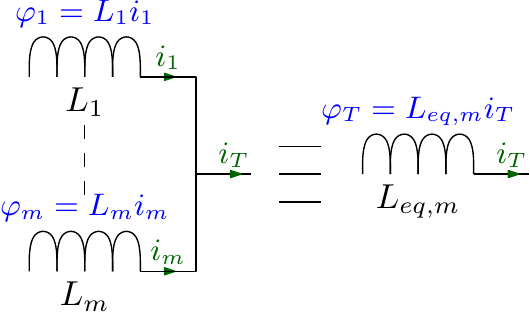}
	\caption{Physical interpretation of $\varphi_T$}
	\label{fig:Leq}
	\vspace{-0.5cm}
\end{figure}

%
We also introduce $\mc C \in \real^{m-1}$ that reflect flux distribution related to the current distribution as
\begin{equation}\label{eq:C}
	\mc C = 	\Gamma_m^\intercal \varphi
\end{equation}
with
\begin{equation}\label{eq:gamma_m}
		\real^{(m-1)\times m}\ni\Gamma_m^\intercal:=G-\begin{bmatrix}\mathbf{0}_{m-1} & \mathbf{I}_{m-1}\end{bmatrix},
\end{equation}
where the $k$-th row of $G\in\real^{(m-1)\times m}$ is
\begin{multline*}
	G_{k}\coloneqq\begin{bmatrix}L_{\text{eq},k}\mathbf{1}_{k}^{\textsf{T}} & \mathbf{0}_{m-k}^{\textsf{T}}\end{bmatrix}\text{diag}\{L\}^{-1}\\
	=L_{\text{eq},k}\begin{bmatrix}\frac{1}{L_1}& \frac{1}{L_2}&\cdots & \frac{1}{L_k}& 0&\cdots &0\end{bmatrix}.
\end{multline*}
Note that the remark named \emph{Expression of $\Gamma_m$} in p.\pageref{rem:ExpressionGamma} gives guidelines for the construction of $\mc C$ for $m=3$.

Next lemma provides dynamical equations in the new coordinates after introducing new input coordinates.
\begin{lem}
	Let $\Phi^{-1}$ and $U^{-1}$ define state and input change of coordinates via
	$z \coloneqq \begin{bmatrix} \mc C^\intercal & \varphi_T & Q \end{bmatrix}^\intercal = \Phi^{-1}x$ and $\begin{bmatrix}\lambda^\intercal & \mu\end{bmatrix}^\intercal=U^{-1}d$ where
	\begin{equation*}
		\Phi^{-1}=\begin{bmatrix}\Gamma_m^\intercal & \0_{m-1}\\\1_m^\intercal\diag L^{-1}L_{\text{eq},m} & 0\\ \0_m^\intercal & 1\end{bmatrix}
	\end{equation*}
	and
	\begin{multline}\label{eq:U}
		U\coloneqq\text{diag}\left\{ E\right\} ^{-1}\begin{bmatrix}\Gamma_{m}^{\textsf{T}}\\
		{\bf 1}_{m}^{\textsc{\textsf{T}}}\text{diag}\left\{ L\right\} ^{-1}L_{\text{eq},m}
		\end{bmatrix}^{-1}\\ \begin{bmatrix}\diag{\tilde E} & \0\\\0 & E_\text{eq}\end{bmatrix},
	\end{multline}
	with $\tilde E\in \real^{m-1}$ and $E_\text{eq}\in \real$ being positive parameters to be chosen.
	Model \eqref{eq:Hamiltmodel} can be equivalently rewritten as
	\begin{multline}\label{eq:modele final}\dot z=\underset{\mc J_z -\mc R_z}{\underbrace{\begin{bmatrix}{\bf 0} & {\bf 0}_{m-1} & {\bf 0}_{m-1}\\
	\mathbf{0}_{m-1}^{\textsf{T}} & 0 & -1\\
	\mathbf{0}_{m-1}^{\textsf{T}} & 1 & -\nicefrac{1}{R}
\end{bmatrix}}}\nabla_z H_z(z)\\+\begin{bmatrix}\text{diag} \left\{\tilde{E}\right\} & \0_{m-1}\\\0_{m-1}^\intercal & E_\text{eq}\\ \mathbf{0}_{m-1}^{\textsf{T}} & 0\end{bmatrix}
		\begin{bmatrix}\lambda\\
			\mu
		\end{bmatrix}.
	\end{multline}
	where $H_z(z)=\frac{1}{2}z^\intercal \mc Q z$ with $\mc Q$ being the following positive definite diagonal matrix
	\begin{equation}\label{eq:calQ}
		\mc Q\coloneqq
		\diag{L_{\mc C},L_{\text{eq},m},C}^{-1}
	\end{equation}
	with $L_\Ca$ a $(m-1)$-dimensional vector defined by
\begin{equation}\label{eq:Lc}
	L_{\Ca,k}\coloneqq L_{\text{eq},k}+L_{k+1},~ \forall k \in \left\{ 1, \cdots,  m-1 \right\}.
\end{equation}
\end{lem}

\begin{proof}
	Since $\Gamma_m^\intercal$ is full-row rank, $\1_m^\intercal\diag{L}^{-1}\1_m=1/L_{\text{eq},m}$ and $\Gamma_m^\intercal \1_m=\0_{m-1}$, we can show that $\Phi$ reads
\begin{equation*}
	\Phi=\begin{bmatrix}
		 \Gamma_m^+&\1_m &\0_m\\
		\0_{m-1}^\intercal & 0 & 1
	\end{bmatrix},
\end{equation*}
where 
\[\Gamma_m^+\coloneqq\frac{\diag{L}}{L_{\text{eq},m}}\Gamma_m\left(\Gamma_m^\intercal \frac{\diag{L}}{L_{\text{eq},m}}\Gamma_m\right)^{-1}.\]
System in hamiltonian form reads
\begin{multline}\label{eq:modelez}
			\dot z=\Phi^{-1}\dot x=\underset{\eqqcolon \mc J_z-\mc R_z}{\underbrace{\Phi^{-1}\left[\mathcal{J}-\mathcal{R}\right]\Phi^{-\intercal}}}\nabla_z H_z(z) \\+\Phi^{-1}\begin{bmatrix}\text{diag}\left\{ E\right\} \\
				{\bf 0}_{m}^{\textsf{T}}
			\end{bmatrix}d
\end{multline}
where the Hamiltonian is such that
\begin{equation*}
	H_z(z)=H(\Phi(z))=\frac{1}{2}z^\intercal \mc Q z 
\end{equation*}
with $\mc Q=\mc Q^\intercal> 0$ being the following block-diagonal matrix
\begin{equation*}
	\mc Q=\begin{bmatrix}(\Gamma_m^+)^\intercal \diag{L}^{-1}\Gamma_m^+ &\0_{m-1}\\
	\0_{m-1}^\intercal & \diag{L_{\text{eq},m},C}^{-1}\\
	 \end{bmatrix}.
\end{equation*}
Using \eqref{eq:gamma_m}, it holds
$
	\left(\Gamma_m^\intercal \diag L \Gamma_m \right)^{-\intercal}\eqqcolon \left([a_{ij}]\right)^{-1}
$
with
	\begin{eqnarray*}
		a_{ii}&=&(L_{\text{eq},i})^2\underset{L_{\text{eq},i}^{-1}}{\underbrace{\left(\sum_{k=1}^i \frac{1}{L_k}\right)}}+L_{i+1}=L_{\Ca,i}\\
		i<j,~a_{ij}&=&L_{\text{eq},i}L_{\text{eq},j}\underset{L_{\text{eq},i}^{-1}}{\underbrace{\left(\sum_{k=1}^{i}\frac{1}{L_k}\right)}}-L_{\text{eq},j}=0.
	\end{eqnarray*}
	Thus $(\Gamma_m^+)^\intercal \diag{L}^{-1}\Gamma_m^+=\diag{L_\Ca}^{-1}$ so that~\eqref{eq:calQ} holds.
Furthermore, $\mc J_z-\mc R_z$ appearing in~\eqref{eq:modelez} reduces to
\begin{eqnarray*}
	\mc J_z-\mc R_z &=&\Phi^{-1}[\mc J-\mc R]\Phi^{-\intercal}\\
	&=&\begin{bmatrix}{\bf 0} & {\bf 0}_{m-1} & {\bf 0}_{m-1}\\
	\mathbf{0}_{m-1}^{\textsf{T}} & 0 & -1\\
	\mathbf{0}_{m-1}^{\textsf{T}} & 1 & -\nicefrac{1}{R}
\end{bmatrix}.
\end{eqnarray*}
Finally, input to state matrix reads
\[
	\Phi^{-1}\begin{bmatrix}\text{diag}\left\{ E\right\} \\
				{\bf 0}_{m}^{\textsf{T}}
			\end{bmatrix}U=\begin{bmatrix}\text{diag} \left\{\tilde{E}, E_\text{eq}\right\}\\ \mathbf{0}_{m}^{\textsf{T}}\end{bmatrix}.\vspace{-0.4cm}
\]
\end{proof}
\begin{rem}[Input change of coordinates]
	After the state transformation that separate flux distribution from voltage and total current dynamics, the input change of coordinates aims to find the input part that only acts on $\Ca$ and the one that only acts on $\varphi_T$ and $Q$.
\end{rem}

From sparsity of matrices in \eqref{eq:modele final}, we are able to separate \eqref{eq:modele final} into two independent subsystems:
\begin{itemize}
\item[$\bullet$] The first one, $\Sigma_\Ca$ is described by the ($m-1$) first lines of \eqref{eq:modele final}:
\begin{equation}\label{eq:systeme1}
	\hspace{-0.5cm}\begin{cases}
		\dot \Ca&=\text{diag}\left\{\tilde{E}\right\} \lambda\\
		H_\Ca(\Ca)&=\frac{1}{2}\Ca^\intercal \diag{L_{\mc C}}^{-1} \Ca,
	\end{cases}
\end{equation}
and corresponds to the dynamics of overall flux repartition among the different branches;
\item[$\bullet$] The second one, $\Sigma_Q$ is described by the two last lines of \eqref{eq:modele final}: 
\begin{equation}
	\label{eq:equivalentbuck}
	\hspace{-0.8cm}\begin{cases}\begin{bmatrix}\dot{\varphi_T}\\
		\dot{Q}
	\end{bmatrix}=\underset{\mc J_Q-\mc R_Q}{\underbrace{\begin{bmatrix}
		 0 & -1\\
		1 & -\nicefrac{1}{R}
	\end{bmatrix}}}\nabla_{\varphi_T,Q} H_Q+\begin{bmatrix}E_{\text{eq}}\\
		0
	\end{bmatrix}\mu\\
	H_Q=\frac{1}{2}\begin{bmatrix}\varphi_T\\ Q\end{bmatrix}^\intercal\diag{L_{\text{eq},m}, C}^{-1}\begin{bmatrix}\varphi_T\\Q\end{bmatrix},\end{cases}
\end{equation}
and governs dynamics of the capacitor charge, through equivalent flux $\varphi_T$ controlled by input $\mu$.
\end{itemize}
This separation is represented on Fig.~\ref{fig:newopenloop}.
\begin{figure}
	\centering
	\includegraphics[width=\columnwidth-1cm]{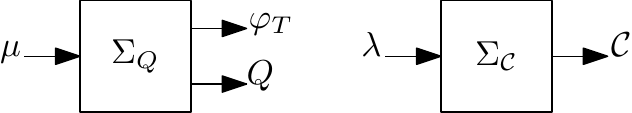}
	\caption{New open-loop model.}
	\label{fig:newopenloop}
	\vspace{-0.5cm}
\end{figure}

\begin{rem}[Fully disconnected subsystems]
Vis-a-vis the main result of \cite{Tregouet2017} where the change of coordinates leads to cascaded subsystems, the Hamiltonian formulation gives rise to a different change of coordinates that fully decouples the two dynamics $\Sigma_\Ca$ and $\Sigma_Q$. Furthermore it induces a diagonal structure of $\Sigma_\Ca$, dynamics of $\mc C$ as well as $H_\Ca$. 
\end{rem}

\begin{rem}[Expression of $\Gamma_m$]\label{rem:ExpressionGamma}
	Note that $\Gamma_m$ could have been defined differently while preserving this aforementioned separation between dynamics. However, we will see in the following that this particular expression leads to a nice structure of the system in the new coordinates and a physical interpretation. Indeed, take the equivalent inductor of the $k$ first branches and the sum of the $k$ first currents, their product corresponds to $G_kx$, the equivalent flux of the $k$ first inductors (see Fig.~\ref{fig:Casimir}).
\begin{figure}
	\centering
	\includegraphics[width=6cm]{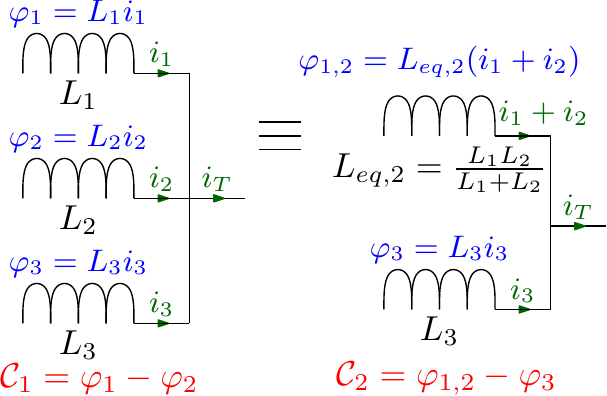}
	\caption{Physical representation of $\Ca$.}
	\label{fig:Casimir}
	\vspace{-0.5cm}
\end{figure}	
As an example, for $k=2$, $G_2x$ reads
\begin{equation*}
	G_2 x=\frac{1}{1/L_1+1/L_2}\left(\frac{\varphi_1}{L_1}+\frac{\varphi_2}{L_2}\right)=L_{\text{eq},2}(i_1+i_2).\vspace{-0.3cm}
\end{equation*}
\end{rem}

\subsection{Circuit theory interpretation}

On the one hand, from its dynamical equation \eqref{eq:equivalentbuck}, subsystem $\Sigma_Q$ can be physically interpreted as average model of buck converter depicted on Fig.~\ref{fig:circuitTheory}~(a) (see \cite{Sira-Ramirez1997}). The equivalent inductor $L_{\text{eq},m}$ of $m$ coils in parallel and the equivalent source $E_\text{eq}$ compose the converter, whereas the input $\mu$ act as a virtual duty cycle. The converter feeds the load $R$ through the capacitor $C$.

On the other hand, $\Sigma_\Ca$ can be seen as ($m-1$) independent electrical circuits $\Sigma_\Ca^k$ represented on Fig.~\ref{fig:circuitTheory}~(b). Each electrical circuit is composed by an electrical source $\tilde E_k$ connected to a virtual coil $L_{\Ca,k}$ via controllable transistors. The circuit is passed through by the following electrical flow $\Ca_k =l_{\Ca,k}\bar i_k$ where $\tilde i_k=\frac{L_{\text{eq},k}}{L_{\Ca,k}} \left(\frac{\varphi_1}{L_1}+\cdots+\frac{\varphi_k}{L_k}-\frac{\varphi_{k+1}}{L_{\text{eq},k}}\right)$.
By the virtual duty cycle $\lambda_k$, we can control the electrical flow $\Ca_k$. If $\Ca_k$ is positive, the $k$ first branches will be favoured to transmit power to the load whereas with $\Ca_k<0$, the branch $k+1$ will take more power than the $k$ first ones. 
Because subsystems $\Sigma_\Ca^k$ are independent from the equivalent buck converter, $\Ca_k$ are only specifying how power flow is allocated among the branches without affecting the power transmitted to the load.
Fig.~\ref{fig:circuitTheory}~(b) depicts that $\Sigma_\Ca^k$ corresponds to the energy transiting between the equivalent inductance $L_{\text{eq},k}$ and $L_{k+1}$ (because $L_{\Ca,k}= L_{\text{eq},k}+L_{k+1}$).
\begin{figure}
	\centering
	   \begin{subfigure}[b]{0.5\textwidth}
        \centering
        \includegraphics[width=7cm]{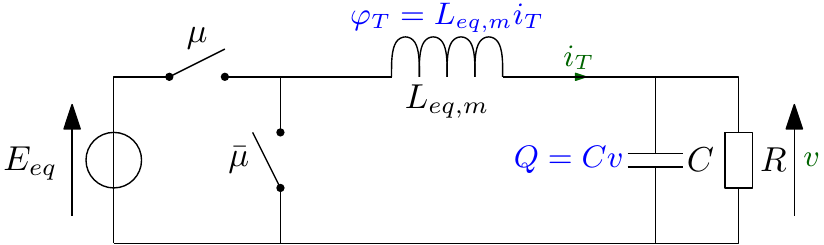}
	\label{fig:circuitTheory}
        \caption{~}
  \end{subfigure}
  \begin{subfigure}[b]{0.5\textwidth}
        \centering
       	\includegraphics[width=5cm]{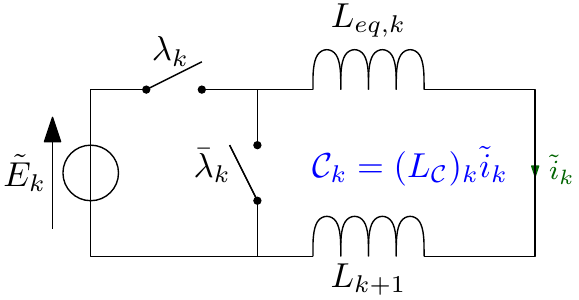}
       	    \caption{~}
  \end{subfigure}
	\caption{Circuit interpretation of (a) $\Sigma_v$ and (b) $\Sigma_\Ca^k$.}
	\label{fig:circuitTheory}
	\vspace{-0.5cm}
\end{figure}
\subsection{Reformulation of Problem 1 in the new coordinates}\label{sec:reformulation}

So far, the contribution is to separate what acts on the charge $Q$, that is $\varphi_T$ and $\mu$ from the {\it free variables} $\Ca$ and inputs $\lambda$.
Let us rewrite the equation of Problem 1 in the variables defined in the previous section:
\begin{equation*}
\begin{array}{ccc}
	\left[\Ca^{*\intercal}, \varphi_T^\star,Q^\star\right]^\intercal\coloneqq& \underset{\Ca,\varphi_T, Q}{\text{argmin}} &J\left(\Phi z\right) \\
	&\text{ s.t. }&\left\{\begin{array}{l}
	Q=Q_r\\
	\dot z=\0_{m+1}.
	\end{array}\right.
\end{array}
\end{equation*}
Knowing the constraint about equilibrium point ($Q=Q_r$), it follows from \eqref{eq:charge} that
\begin{equation*}
	\sum_{k=1}^{m}\frac{\varphi_k^\star}{L_k}=\frac{Q_r}{RC} \Leftrightarrow \varphi_T^\star=\frac{L_{\text{eq},m}Q_r}{RC}.
\end{equation*}
Thus, the constraints impose the asymptotic values of variables related to the equivalent buck (i.e. $\varphi_T$ and $Q$) and they are no longer decision variable of the optimization problem. Hence, optimization subproblem of Problem 1 reduces to:
\begin{equation}\label{eq:problembis}
	\Ca^\star\coloneqq \underset{\Ca}{\text{argmin}} ~J_z\left(\Ca,\varphi_T^\star\right)\;\;\text{s.t.}\;\;
	\dot \Ca=\0_{m-1},
\end{equation}
where $J_z$ and $\varphi_T^\star$ read
\begin{align*}
J_z&:(\Ca,\varphi_T)\mapsto J\left(\Phi \begin{bmatrix}\Ca^\intercal& \varphi_T&Q_r\end{bmatrix}^\intercal\right) \\
\varphi_T^\star&\coloneqq L_{\text{eq},m}Q_r/(RC)
\end{align*}

The separation in two blocs confines dynamics of $\varphi_T$ and $Q$ in a single subsystem. Independently of cost function $J_z$, those two variables must converges to $\varphi_T^\star$ and $Q_r$ respectively to solve Problem 1. In fact, this dynamics refers to voltage regulation objective. Second subsystem $\Sigma_\Ca$ of variables $\Ca$ must converge to an optimal value of cost function $J_z$. This dynamics refers to the optimization of power flow repartition among all the converters, for a chosen total power transmitted to the load.

\begin{rem}[Casimir functions]
	In this section, we gave a change of coordinates that decompose the system in two parts to separate voltage regulation from flow distribution. In fact, this change of coordinates is closely related to the presence of Casimir functions.
	A Casimir function $\mc C:\real^n\rightarrow \real$ is a conservative function regardless of the Hamiltonian (see \cite[p.87]{Schaft2000}). It expresses the dynamical invariants and is a solution of
\begin{equation}\label{eq:Casimir}
	\frac{\text{d}\mc C}{\text{d}t}=0~ \Leftrightarrow ~ \frac{\partial \mc C}{\partial x}\dot x=0.
\end{equation}
As the property of $\mc C$ holds for all $H$, by including the model of an autonomous Hamiltonian system into \eqref{eq:Casimir} we obtain the following relation
\begin{equation}\label{eq:dynamical invariance2}
	\frac{\partial\Ca}{\partial x}(x)\left[\mathcal{J}(x)-\mathcal{R}(x)\right]=0.
\end{equation}
From \cite[p 87]{Schaft2000}, we know that Casimir functions can be used in order to achieve a change of coordinates that isolate those functions from the rest of the state. Since $\Gamma_m^\intercal \1_m=\0$ and
\[
	\frac{\partial \Ca}{\partial x} = \begin{bmatrix}
		\Gamma_m^\intercal	&	\0_{m}	\\
	\end{bmatrix}
\] 
with $\mc C$ given in \eqref{eq:C}, it is clear that \eqref{eq:dynamical invariance2} is satisfied, so that every entry of vector $\mc C$ is a Casimir function. 
\end{rem}

\section{Control design}

Starting from the case where the load is known (Subsection~\ref{sec:Rconnue}), this section then gives a load-independent solution to Problem~1 (see Subsection~\ref{sec:unknown}).

\subsection{Control design with known $R$\label{sec:Rconnue}}

Firstly, we want to solve Problem 1 when $R$ is known.

\subsubsection{Control of equivalent Buck $\Sigma_Q$}

Let us first design a controller for $\Sigma_Q$ that ensure that $Q\rightarrow Q^\star=Q_r$ and $\varphi_T\rightarrow \varphi_T^\star$. Here, the controller is based on the well known Interconnection and Damping Assignment Passivity-Based Control (IDA-PBC) procedure introduced in \cite{Ortega2002}. The desired closed-loop behaviour is written as the following Hamiltonian system
\begin{equation}\label{eq:closed-loop1}
	\begin{bmatrix}\dot \varphi_T\\\dot Q\end{bmatrix}=[\mc J_Q^d-\mc R_Q^d]\nabla_{\varphi_T,Q}H_Q^d(\varphi_T,Q),
\end{equation}
where
\begin{align*}
	\mc J_Q^d &\coloneqq\begin{bmatrix}0 & -1 \\1 & 0\end{bmatrix}, \qquad\mc R_Q^d \coloneqq\begin{bmatrix}k_\mu &0\\0 & 1/R\end{bmatrix},\\
	H_Q^d&\coloneqq\frac{1}{2L_{\text{eq},m}}\left(\varphi_T-\frac{L_{\text{eq},m}Q_r}{RC}\right)^2+\frac{1}{2C}\left(Q-Q_r\right)^2.
\end{align*}
Minimum of $H_Q^d$ is reached for $Q=Q_r$ and $\varphi_T=\varphi_T^\star$. To obtain this closed-loop, we apply the following state feedback controller
\begin{equation}\label{eq:mu1}
	\mu=-\frac{k_\mu}{E_\text{eq}}\frac{1}{L_{\text{eq},m}}\left(\varphi_T-\frac{L_{\text{eq},m}Q_r}{RC}\right)+\frac{Q_r}{C}.
\end{equation}

\begin{prop}
	Equilibrium point $(\varphi_T^\star,Q_r)$ of closed-loop \eqref{eq:closed-loop1} is GAS for any $R>0$ if $k_\mu>0$.
\end{prop}

\begin{proof}
	$H_Q^d$ is a strictly convex function (quadratic) which is minimum at $(\varphi_T^\star,Q_r)$. Furthermore, $k_\mu>0$ and $R>0$ ensure $\mc R_d>0$, such that
	\[\forall\varphi_T,Q,\quad\dot H_d=-(\nabla_{\varphi_T,Q}H_Q^d)^\intercal \mc R_d \nabla_{\varphi_T,Q}H_Q^d<0.\]
Thus $H_d$ is a Lyapunov function of \eqref{eq:closed-loop1} and the closed-loop converges globally and asymptotically to the state value where $H_d$ is minimum, that is $(\varphi_T^\star,Q_r)$
\end{proof}

\subsubsection{Control of $\Sigma_\Ca$}

Let us define the control law
\begin{equation}\label{eq:lambda1}
	\lambda=-\diag{\tilde E}^{-1}K_\lambda\nabla_\Ca J_z\left( \Ca, \varphi_T^\star\right),
\end{equation}
where $K_\lambda=K_\lambda^\intercal\in \real^{(m-1)\times(m-1)}$. In such a case, the closed-loop dynamics reads
\begin{equation}\label{eq:closedloop2}
	\dot \Ca=-K_\lambda\nabla_\Ca  J_z\left(\Ca, \varphi_T^\star\right).
\end{equation}

\begin{asm} \label{asm:Jz}
Map $J_z$ is (i) continuously differentiable and such that (ii) for all $\varphi_T >0$, map $J_z(\cdot,\varphi_T)$ is strictly convex\footnote{Note that for the experimentations described in Section~VI, $J$ reflects the converter losses. Relevant from an engineering view point, this cost function is strictly convex.} and admits a minimum.
\end{asm}

\begin{prop}
	If Assumption~1 hold and $K_\lambda$ is a positive definite matrix, then the closed-loop \eqref{eq:closedloop2} converges globally and asymptotically to the value where the cost function $J_z(\cdot,\varphi_T^\star)$ is minimum.
\end{prop}

\begin{proof}
	 By convexity, minimum of $J_z(\cdot,\varphi_T)$ is unique for all $R>0$. Indeed, assumption on $J_z(\cdot,\varphi_T)$ (existence of minimum and strict convexity) applies for all $\varphi_T>0$ and, in turn, for all $\varphi_T^\star (R)$ as soon as $R>0$ (see definition of $\varphi_T^\star (R)$). As
	\[\dot{J}_z (\Ca,\varphi_T^\star)=-(\nabla_\Ca J_z)^\intercal K_\lambda\nabla_\Ca J_z<0\]
if $K_\lambda>0$, $J_z(\Ca,\varphi_T^\star)$ is a Lyapunov function for the closed-loop \eqref{eq:closedloop2}. The state value where $J_z(\Ca,\varphi_T^\star)$ is minimum is then a GAS equilibrium point of this closed-loop. 
\end{proof}

\subsubsection{Solution of Problem 1}

Finally, resulting from control of both $\Sigma_Q$ and $\Sigma_\Ca$, we enunciate the following theorem:

\begin{theorem}
	If $k_\mu>0$, $K_\lambda \succ 0$ and Assumption~1 hold, then load dependent control law defined by \eqref{eq:mu1}, \eqref{eq:lambda1} and $d=U\begin{bmatrix}\lambda^\intercal &\mu \end{bmatrix}^\intercal$ with $U$ giving by \eqref{eq:U} solves Problem~1 for all $R>0$.
\end{theorem}

\begin{proof}
	From Proposition 1, we know that the constraint $Q=Q_r$ of Problem 1 is fulfilled. As explained in Section~\ref{sec:reformulation}, set of decision variables of optimization problem of Problem 1 reduces to $\Ca$. Since controller $\eqref{eq:lambda1}$ ensure that $J_z(\Ca,\varphi_T^\star)$ is minimum with respect to $\Ca$ (see Proposition 2), then Problem 1 is solved.
\end{proof}

\subsection{Control design with unknown $R$\label{sec:unknown}}

Considering that $R$ is unknown leads to two problems:
\begin{itemize}
\item[$\bullet$] For the control of $\Sigma_Q$, the previous section shows that a classical IDA-PBC controller requires the load value (eq. \eqref{eq:mu1} relies on $R$) because steady-state value depends on $R$;
\item[$\bullet$] The cost function $J$ potentially depends on $\varphi_T^\star$, the equilibrium value of $\varphi_T$ and therefore on $R$. As the controller of $\Sigma_\Ca$ relies on the gradient of $J(\Ca,\varphi_T^\star)$, it also depends on $R$.
\end{itemize}
One of the consequences is that a unilateral interconnection (variable $\varphi_T$) is introduced to estimate the load magnitude for each controllers. Hence, the desired closed loop model shall adopt cascaded form depicted on Fig.~\ref{fig:cascade} where $\Sigma_Q^d$ and $\Sigma_\Ca^d$ refer to desired closed loop of $\Sigma_Q$ and $\Sigma_\Ca$.
\begin{figure}
	\centering
	\includegraphics[width=6cm]{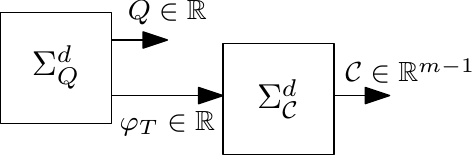}
	\caption{Cascaded interconnection}
	\label{fig:cascade}
	\vspace{-0.5cm}
\end{figure}

One way to deal with this cascaded form is to use main result of \cite{Sontag1989}. It comes out that if the upper subsystem (for us $\Sigma_Q^d$) has a GAS equilibrium and if the lower subsystem (for us $\Sigma_\Ca^d$) has a {\bf0-GAS} equilibrium (i.e. GAS when the input is identically 0), it simply requires that trajectories are bounded for the equilibrium of the entire system to be GAS.

\subsubsection{PID-like control of equivalent Buck $\Sigma_Q$}

In this subsection, we are interested in controlling the equivalent Buck converter \eqref{eq:equivalentbuck} by getting rid of the load value dependance. Robust energy shaping control methods (see \cite{Romero2013}) take into account this kind of issues. They consist in interconnecting the system with a PCH controller. This controller acts as an integrator, hence it is also named PI-like control.

In our case, the regulated output (state $Q$) does not corresponds to the passive output of the PCH system \eqref{eq:equivalentbuck} (state $\varphi_T$). As a result, with a classical interconnection to an integral controller, we are not able to design a $R$-independent control law that stabilize the system at $Q^\star=Q_r$. Hence we resort to the interconnection of PCH systems on non-passive outputs. 
In \cite{Ortega2012}, an approach is presented in order to deal with the case of integral control on non-passive outputs in the Hamiltonian form. 

\begin{rem}[Unknown parameter]
Dealing with unknown parameter $R$ can be cast into a problem of
disturbance rejection by considering that deviation of $R$ with respect
to its nominal value $R_0$ is a perturbation to be rejected by the
closed-loop system. Letting $\Delta G$ be defined by $1/R=1/R_0+\Delta
G$, matrices $\mc J_Q-\mc R_Q$  of system \eqref{eq:equivalentbuck} reads
\begin{equation*}
     \mc J_Q-\mc R_Q=\begin{bmatrix}
          0 & -1\\
         1 & -\left(\nicefrac{1}{R_0}+\Delta G\right)
     \end{bmatrix},
\end{equation*}
As a result, and in stark contrast with methodology proposed in
\cite{Ortega2012} which focuses on \emph{constant exogeneous}
disturbance, we have to tackle disturbance $\Delta G$ which is
multiplied by $\nabla H_Q$ in the expression of derivative of
$(\varphi_T,Q)$ and is, in turn, (linearly) state-dependent. Yet results
of \cite{Ortega2012} paves the way in construction of robust control law
which is proposed in this paper.
\end{rem}

	Let $C_\mu$ be the control law
\begin{equation}\label{eq:controllerSigma}
	\begin{cases}
		\dot \xi &=k_i\frac{Q-Q_r}{C}\\
		\mu&=-\frac{1}{E_\text{eq}}\left(k_d\frac{\varphi_T}{L_{\text{eq},m}}+k_d\xi+L_{\text{eq},m}k_i\frac{Q-Q_r}{C}\right)
	\end{cases}
\end{equation}
with $k_d,k_i\in \real$ be the integral controller on the non-passive output of system \eqref{eq:equivalentbuck}.

\begin{prop}
	The closed-loop of \eqref{eq:equivalentbuck} and \eqref{eq:controllerSigma} converges globally and asymptotically to some equilibrium point for which $Q$ equals $Q_r$ for all $R>0$ if $k_d>0$ and $k_i>0$.
\end{prop}

\begin{proof}
From \cite{Ortega2012} one can prove that the integral controller \eqref{eq:controllerSigma} and system \eqref{eq:equivalentbuck} can be written in the following coordinates
\begin{equation}\label{eq:chi}
	\chi\coloneqq\begin{bmatrix}
		1 & 0 & L_{\text{eq},m}\\
		0&1&0\\
		0&0&1
	\end{bmatrix} \begin{bmatrix}
		\varphi_T-\frac{L_{\text{eq},m}}{RC}Q_r\\
		Q-Q_r\\
		\xi-\left(-\frac{1}{k_d}-\frac{1}{R}\right)\frac{Q_r}{C}
	\end{bmatrix}
\end{equation}
as the Hamiltonian system:
\begin{equation}\label{eq:modelz}
	\dot \chi=\begin{bmatrix}
	-k_d & -1 & 0\\
	1 & -1/R & -k_i\\
	0 & k_i & 0
	\end{bmatrix}\nabla_\chi H_d (\chi),
\end{equation}
where
\begin{equation}\label{eq:Hd}
	H_d(\chi)=\frac{1}{2}\chi^\intercal \diag{\left[L_{\text{eq},m}, C, k_i\right]}^{-1}\chi.
\end{equation}

	By specifying $k_i>0$, we ensure that \eqref{eq:Hd} is a strictly convex function over $\real^3$. Furthermore, it directly follows from \eqref{eq:Hd} that it admits a (unique) minimum at $\chi^\star =\0_3$. 
	In addition to that, $\dot H_d=\left(\nabla_\chi H_d(\chi)\right)^\intercal \mc R_d \nabla_\chi H_d(\chi)\leq 0$ with $\mc R_d\coloneqq\diag{[k_d, 1/R, 0]}$. Knowing that the largest invariant contained in
	\[
		\mc S\coloneqq\left\{\chi \in \real^3|\dot H_d=0\right\} = \left\{\chi \in \real^3|\chi_1=\chi_2=0\right\}
	\]
	is $\{\0_3\}$ because 
	\begin{eqnarray*}
		\dot \chi=\begin{bmatrix}
	-k_d & -1 & 0\\
	1 & -1/R & -k_i\\
	0 & k_i & 0
	\end{bmatrix}\nabla_\chi H_d (0,0,\chi_3)\subseteq\mc S\\
	\Leftrightarrow (0,-k_i\chi_3,0)\subseteq \mc S \quad \Leftrightarrow\quad\chi_3=0,
	\end{eqnarray*}
	it follows from LaSalle's Invariance Principle (see \cite{Khalil2002}) that $\chi^\star=\0_3$ is GAS. Furthermore, it is clear from \eqref{eq:chi} that when $\chi=\0_3$, it holds $Q=Q_r$.
\end{proof}

\begin{rem}[Controller properties]
	Gain $k_d$ refers to the feedback of the flux variable $\varphi_T$ whereas $k_i$ have an influence on the integral action of the charge $Q$. We also notice that: (i) if $k_d=0$, then there is no integral action of the controller and (ii) $k_i$ also have a proportional feedback action on the charge $Q$.
\end{rem}

\subsubsection{Control of $\Sigma_\Ca$}

In general, the cost function $J_z(\cdot,\varphi_T^\star)$ depends on the load $R$ via $\varphi_T^\star$ which is unknown. This is the reason why, we consider the following control law of $\Sigma_\Ca$, in place of \eqref{eq:lambda1},
\begin{equation}\label{eq:lambda2}
	C_\lambda:\quad \lambda=-\diag{\tilde E}^{-1}K_\lambda\nabla_\Ca J_z\left(\Ca, \varphi_T\right)
\end{equation}
which leads to the closed loop
\begin{equation*}
	\Sigma_\Ca^d: \quad\dot \Ca=-K_\lambda\nabla_\Ca  J_z\left(\Ca, \varphi_T \right).
\end{equation*}
As already discussed at the beginning of Subsection~\ref{sec:unknown}, $\varphi_T$ is now an input of $\Sigma_\Ca^d$ and we will see here stability property of this subsystem when this input is at rest, that is when $\varphi_T=\varphi_T^\star$. In such a case, we recover subsystem already considered in Proposition~2, so that the following result can be established in the same way, since $\varphi_T^\star >0$ is arbitrary in Proposition~2.

\begin{prop}
Assume $K_\lambda \succ 0$ and Assumption~1 holds. Then for any strictly positive and constant $\varphi_T$, set point $\textstyle\argmin_\Ca J_z(\Ca,\varphi_T)$ is GAS.
\end{prop}

\subsubsection{Solution of Problem 1}

Making use of individual controllers of both $\Sigma_Q$ and $\Sigma_\Ca$, we now establish main result of this paper.

\begin{theorem}
Let $\Ca^\star(\varphi_T)\coloneqq\textstyle\argmin_\Ca J_z(\Ca,\varphi_T)$ and define maps $h:\real^{m-1}\times\real\rightarrow\real^{m-1}$ and $W:\real^{m-1}\rightarrow\real$ as follows:
\begin{align*}
h(\tilde\Ca,\tilde\varphi_T)&  \coloneqq  \nabla_\Ca  J_z (\tilde\Ca+\Ca^\star, \tilde\varphi_T+\varphi_T^\star)\\& \quad\qquad- \nabla_\Ca  J_z (\tilde \Ca+\Ca^\star, \varphi_T^\star) \\
W(\tilde\Ca)\coloneqq & J_z(\tilde\Ca+\Ca^\star,\varphi_T^\star)-J_z(\Ca^\star,\varphi_T^\star).
\end{align*}
Assume that the two following facts hold:
\begin{itemize}
\item[F1)] There exists two strictly increasing functions $G_{1,2}$ which are null and differentiable at the origin and such that
$
\|h(\tilde\Ca,\tilde\varphi_T)\| \leq G_1 ( |\tilde\varphi_T |) + G_2 (|\tilde\varphi_T|) \|\tilde\Ca\|
$

\item[F2)] There exist positive constants $c$ and $k$ such that $\|\tilde\Ca\| >c$ implies
$
\left\| \nabla W(\tilde\Ca)\right\| \|\tilde\Ca\| \leq k  W(\tilde\Ca)
$
\end{itemize}
Then, load independent control law defined by \eqref{eq:controllerSigma}, \eqref{eq:lambda2} and $d=U\begin{bmatrix}\lambda^\intercal &\mu \end{bmatrix}^\intercal$ with $U$ given by \eqref{eq:U} solves Problem 1 if $K_\lambda=K_\lambda^\intercal \succ 0$, $k_d>0$, $k_i>0$ and Assumption~1 holds.
\end{theorem}

\begin{proof}
From Proposition 3, we know that for any $R>0$, there exists $\xi^\star(R)$ such that $(\varphi_T^\star(R),Q_r,\xi^\star(R))$ is an GAS equilibrium for $\Sigma_Q^d$. From Proposition 4, we have that $\Ca^\star (\varphi_T)=\textstyle\argmin_\Ca J_z(\Ca,\varphi_T)$ is a GAS equilibrium of $\Sigma_\Ca^d$ for any constant $\varphi_T>0$. From \cite{Sontag1989}, in such a case, boundedness of trajectories implies global asymptotic stability of $(\Ca^\star (\varphi_T^\star(R)),\varphi_T^\star(R),Q_r,\xi^\star(R))$ for the whole closed loop  system for all $R>0$. To prove boundedness property, first note that existence of a GAS equilibrium for $\Sigma_Q^d$ implies boundedness of the $(\varphi_T,Q,\xi)$ substate. Then, observe that dynamics of $\Sigma_\Ca^d$ can be reformulated as follows using relative coordinates $\tilde \varphi_T=\varphi_T-\varphi_T^\star$ and $\tilde \Ca=\Ca-\Ca^\star$:
\begin{equation*}
\dot{\tilde\Ca} = -K_\lambda\nabla_\Ca  J_z\left(\tilde\Ca+\Ca^\star,\varphi_T^\star\right) - K_\lambda h(\tilde\Ca,\tilde\varphi_T).
\end{equation*}
In such case, and whenever Assumption~1, F1) and F2) hold, all the hypothesis for the applicability of \cite[Lemma~1]{Jankovic1996} are satisfied 
which proves that $W(\Ca(t))$ remains bounded which, in turn, proves boundedness of $\Ca(t)$ since $\Ca \mapsto J_z(\Ca,\varphi_T)$ is radially unbounded for all $\varphi_T>0$, due to its convexity and the existence of a minimum.
\end{proof}

\begin{rem}[About F1) and F2)]
Hypothesis F1) imposes linear growth with respect to $\Ca$ of $h$, the coupling term from $\Sigma_Q^d$ to $\Sigma_{\mathcal C}^d$. This prevents finite time escape of $\Ca$ due to this interaction between subsystems \cite{Jankovic1996}.
Regarding F2), \cite[Lemma~2]{Jankovic1996} implies that this fact is satisfied if $J_z(\cdot,\varphi_T^\star)$ is polynomial for any $\varphi_T^\star>0$, in addition of fulfilling requirements of Assumption~1.
\end{rem}

{\color{blue}

\section{Discussions about Robustness}

The considered converter aims regulating output voltage in spite of \emph{unknown} load variation. Enjoyed by controller of Theorem~2, robustness with respect to load magnitude is indeed an essential feature for any control scheme dealing with realistic scenarii. It has been theoretically proved that this controller achieves exact voltage regulation for every $R>0$ and is optimal at the steady-state in terms of flux repartition.

Let us now provide insights about how other model parameter uncertainties might affect closed-loop stability and performance.

\subsection{Uncertainties on capacitor}

In most practical case, proposed controller is structurally robust to arbitrary large uncertainty with respect to capacitor magnitude $C$. To see it, first observe that this controller is built up from $\Phi^{-1}$, $C_\mu$, $C_\lambda$ and $U$. It is clear that $\Phi^{-1}$ and $U$ are independent from $C$. Furthermore, even if $C_\mu$ depends on $Q/C$ and $Q_r/C$, voltage $v=Q/C$ is a measurable variable. This removes dependency from the capacitor $C$ since $v=Q/C$ and $v_r=Q_r/C$ hold. With the assumption that $\nabla_\Ca J_z(\Ca,\varphi_T)$ is independent of $C$, it should be clear that the knowledge of $C$ is not required to implement this controller.

\subsection{Uncertainties on inductors}

When implementing proposed controller, defect of electrical components might be troublesome. In order to anticipate this issue, dynamics of non ideal inductors is captured by including Equivalent Series Resistance (ESR) vector $r\in\real^m$ to the model, so that $-r_k \varphi_k /L_k$ is added on the right-hand side of \eqref{eq:flux}. If resistance magnitudes are known, it suffices to perform a ``pre-feedback'' of the form $d_k=r_k\varphi_k /(L_k E_k)+\tilde d_k$ to recover original version of \eqref{eq:flux} with $\tilde d_k$ in place of $d_k$. Then, controller mapping $x$ to $\tilde d_k$ can be designed as in previous section.

In almost all practical cases, entries of vector $r$ as well as inductor magnitudes suffer from uncertainties, though. The actual controller has then to possess robustness properties against the (inevitably) inaccurate pre-feedback. Further, the use of nominal values (instead of actual values) to perform change of coordinates might affect independence between $\Sigma_{\mc C}$ and $\Sigma_Q$.

However, it is expected that deviation of $L$ and $r$ with respect to their nominal values is small. This is expected not to comprise stability in most practical situations: Next section provides experimental evidence that the proposed controller can behave properly even without pre-feedback.\footnote{In practice, entries of state-matrix of closed-loop system is continuous with respect to $L$ and $r$ at their nominal values. As a result, this proves that eigenvalues of this matrix must vary continuously with the entries of $L$ and $r$. This, in turn, demonstrates that if asymptotic stability is ensured for the nominal case, then their exists a neighborhood, in the parametric space associated with $L$ and $r$, in which asymptotic stability is preserved.
} 
Furthermore, observe that if stability is preserved, then integral action in \eqref{eq:controllerSigma} will compensate for uncertainties so that voltage will eventually converge to its reference. Note that flux distribution $\mc C$ might converge to a slightly different value than the optimal one, though. This means that if the last secondary objective could be compromised, primary objective of voltage regulation is expected to be fulfilled in the robust case.

\begin{rem}[Casimir functions (continued)]
	 It can be proved that including ESR of inductors into the model prevents existence of Casimir functions. This suggests that implementation of the pre-feedback can be interpreted as a way to make Casimir functions re-appear. We claim that matrix associated with this pre-feedback is closely related to so-called ``friend'' of some invariant subspace of the state-space for which $Q$ is identically zero (see \cite{Wonham2012}). Deeper investigation on this topic is a current line of research.
\end{rem}

}

\section{Experimentations}

In this section, the implementation of the proposed approach is illustrated by experimental results.
\begin{figure}
	\centering
	\includegraphics[width=\columnwidth]{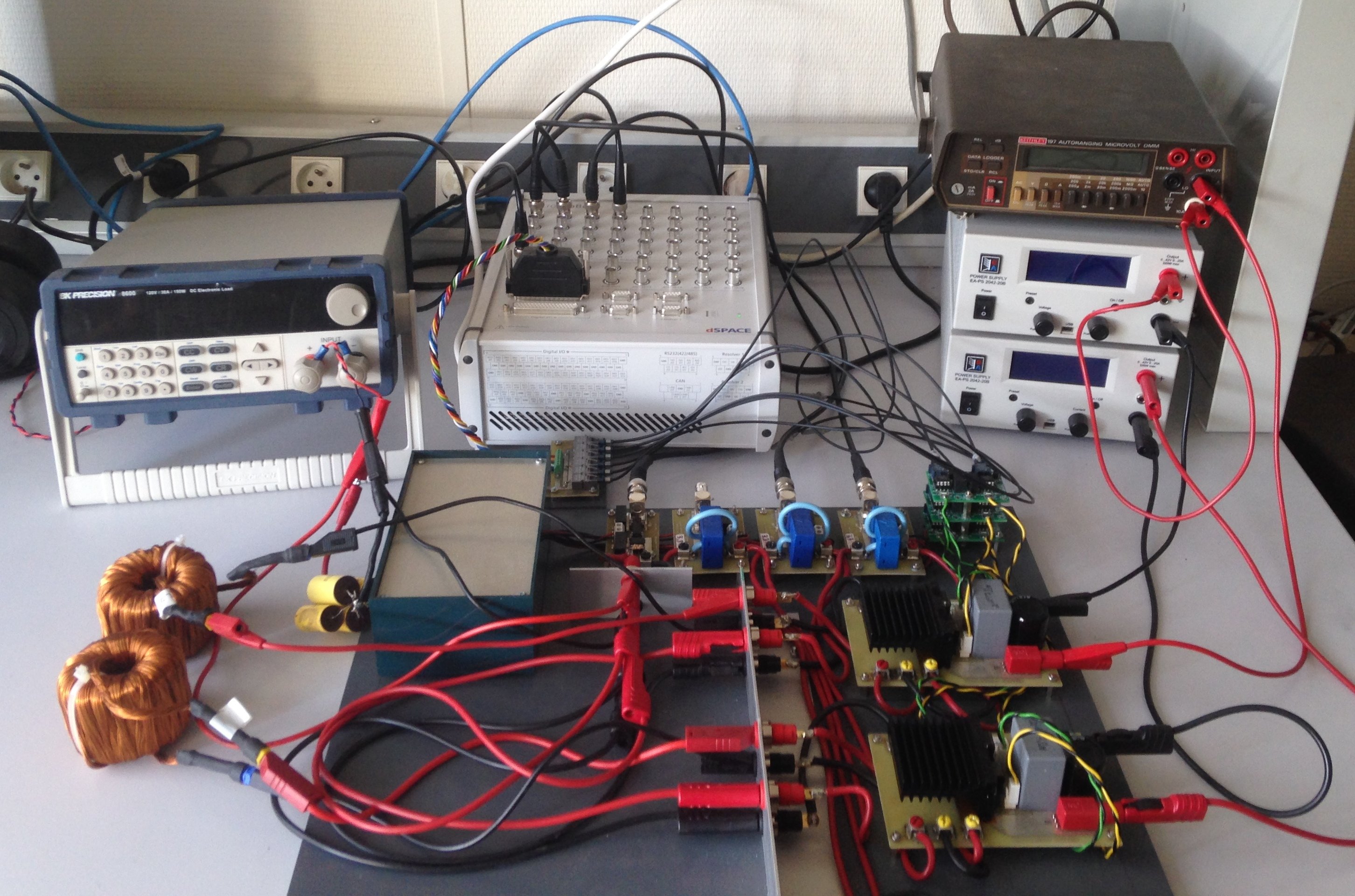}
	\caption{Experimental test-bench}
	\label{fig:banc}
	\vspace{-0.5cm}
\end{figure}

The experimental setup, depicted by Fig.\ref{fig:banc}, is composed of $2$ heterogeneous buck converters ($m=2$) in the sense that inductors, as well as transistors are different. The second converter is designed in such a way that its passive elements have lower quality but the switches have a better efficiency. This means that for low power, the use of converter 2 is preferable whereas converter 1 should have priority for high power. See \cite{Delpoux2018} for detailed discussion on this feature. 

The controller hardware is a dSpace MicroLabBox. For any $R$, control objectives are (i) regulate charge $Q$ at the reference $Q_r=264$ mC which corresponds to a voltage reference $v_r=12$ V and (ii) impose optimal flux distribution through the converters with respect to a cost function $J$. The load variations are performed by a DC electronic load BK Precision 8600 series with a maximum power of 150W and controlled by the dSpace board.

Bench parameters are the followings. The two input voltages are such that $E_1=E_2=24$ V. The switching frequency is chosen as $f_s=20$kHz whereas the sampling frequency is $f_e=10kHz$. All the transistors are MOSFET. For the first converter, their references are STP31510F7 and for the second, STP30NF10. Inductor of the first converter is $K_1=2.83$mH and for the second is $L_2=1.3$mH. Finally the output capacitor value is $C=22mF$.

%
%
%
%

For both experiments, we apply the control law \eqref{eq:controllerSigma}, \eqref{eq:lambda2} and the changes of coordinates $\Phi^{-1}$ and $U$ with the following parameters value:
\[
	k_d=1,~ k_i=10  \text{ and } K_\lambda=0.1,
\]
which comply with statements of Theorem 2.

\subsection{Experiment 1: Decomposition Highlighting}

\subsubsection{Cost function}
For this experiment, the cost function is purely academic and defined as
\[J(\varphi)=\frac{1}{2}(\varphi_1-\varphi_2-\Ca^\star)^2,\]
where $\Ca^\star\in\real$ is a parameter. $J$ has been chosen for its straightforward expression in the $z$-coordinates since
\[J_z(\Ca)=\frac{1}{2}(\Ca-\Ca^\star)^2.\]
Therefore, asymptotic value of $\Ca$ is $\Ca^\star$ for which minimum of $J_z$ is reached. Note that $J_z$ trivially satisfied Assumption~1.

\subsubsection{Experiment environment}

The experiment is divided into three phases:
\begin{itemize}
\item[$\bullet$] Phase \ding{192}: $t\in[0;1[$ s. Initially at 0 C, at $t=0~$s the charge reference $Q_r$ is set to 264 mC for a load value of $R=20~\Omega$. Parameter $\Ca^\star$ is set at 0 Wb.
\item[$\bullet$] Phase \ding{193}: $t\in[1;2[$ s. At $t=1$ s, load magnitude switches to $R=5~\Omega$ and $\Ca^\star$ is still at 0$~$Wb.
\item[$\bullet$] Phase \ding{194}: $t\in[2;3]$ s. At $t=2$ s, $\Ca^\star$ is set at $5\times10^{-3}$ Wb.
\end{itemize}

Because of decoupling of $\Sigma$ in $\Sigma_Q$ and $\Sigma_\Ca$, value of $\Ca$ is expected to remain unchanged when Phase \ding{193} occurs whereas value of $Q$ and $\varphi_T$ are expected to remain unchanged when Phase \ding{194} occurs.

\subsubsection{Results}

Results of Experiment 1 are given by Fig.~\ref{fig:expe1}. Subplot 1 depicts charge $Q$ with the reference $Q_r$, Subplot 2 depicts fluxes through $L_1$ and $L_2$ while Subplot 3 displays $\Ca$ with the reference $\Ca^\star$.

\begin{figure}
	\centering
	\includegraphics[width=\columnwidth]{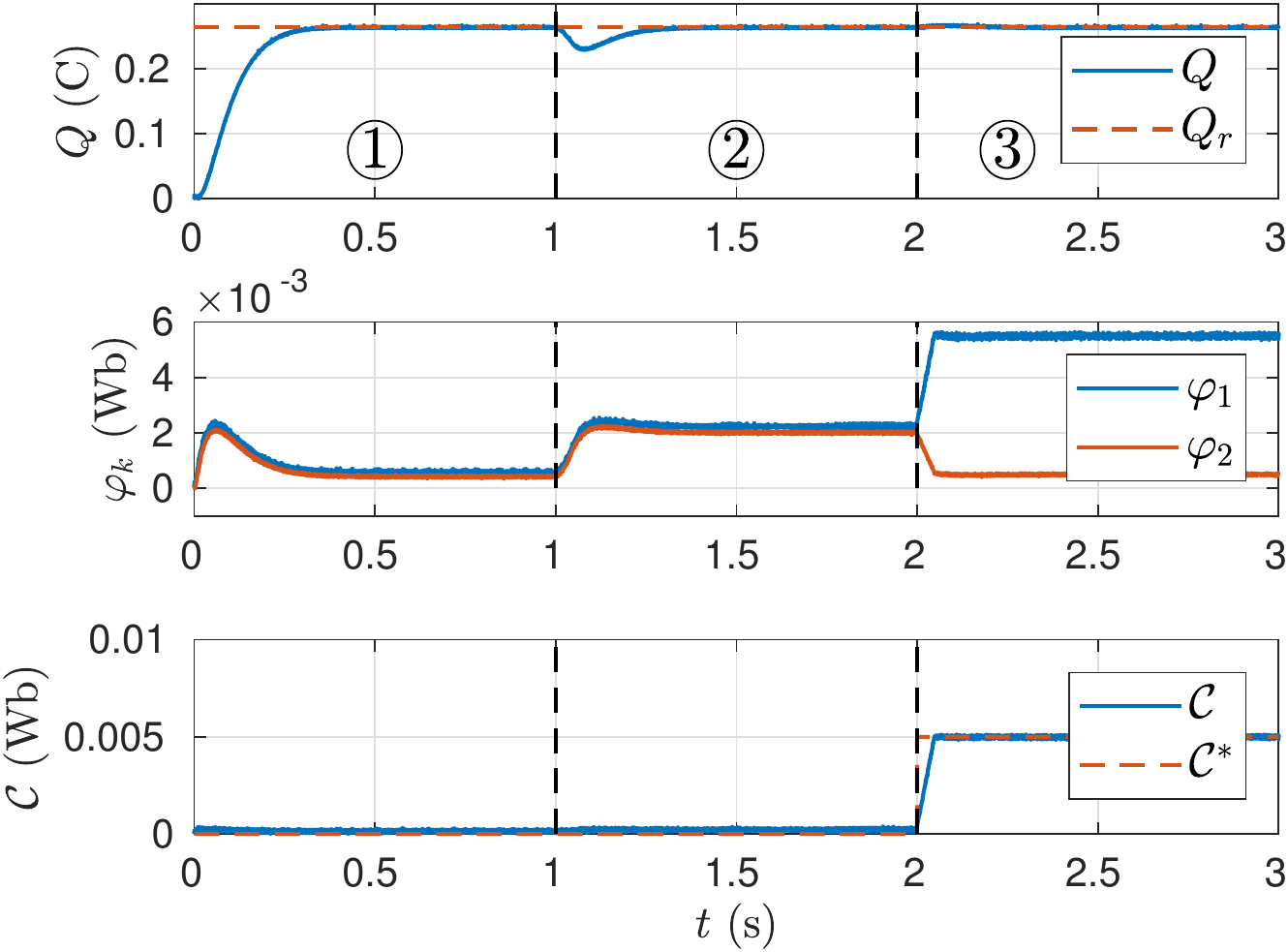}
	\caption{Time results for Exp. 1.}
	\vspace{-0.5cm}
	\label{fig:expe1}
\end{figure}
We see on Fig.~\ref{fig:expe1} that when the load changes value ($t=1$~s), the flow distribution $\Ca$ is almost not impacted. In the same way, when the flow distribution is varying ($t=2$ s), the load charge $Q$ is almost not impacted. An extremely small overshoot can be observed on $Q$. This might come from uncertainties on $L$ and ESR of inductors. Magnitude of this overshoot proves that $\Sigma_Q$ and $\Sigma_\Ca$ are almost fully (robustly) disconnected as shown by Fig.~\ref{fig:newopenloop}. This validates robustness of the approach.
\subsection{Experiment 2: minimization of power losses}

Experiment 2 aims providing a meaningful practical application of paper result. Minimization of power losses is considered for defining the cost function and an unknown load variation is taken into account.

\subsubsection{Experiment environment}\label{experimentEnv}

The experiment is divided into two phases:
\begin{itemize}
\item[$\bullet$] Phase \ding{192}: $t\in[0;1[$ s. Initially at 0 C, at $t=0$ s the charge reference $Q_r$ is set to 264 mC. During this phase, the load value is $R=20~\Omega$.
\item[$\bullet$] Phase \ding{193}: $t\in[1;2]$ s. At $t=1$ s, magnitude of the load is instantly changed to $R=5~\Omega$.
\end{itemize}

\subsubsection{Cost function}

In \cite{Tregouet2016}, it is stated that power losses in $k$-th converter can be expressed as the following quadratic function in terms of converter current $i_k$:
\begin{equation*}\label{eq:powerlosses}
	p_k(i_k)=r_{1,k}i_k^2+r_{2,k}i_k,
\end{equation*}
where $r_{1,k}$ and $r_{2,k}$ are constants depending on electrical components and come from a finer model than \eqref{eq:model} (see \cite{Tregouet2016} for numerical values).
We consider minimization of overall losses, i.e. the sum of $p_k$, as a secondary objective, so that the cost function reads
\[
	J(\varphi)=\sum_{k=1}^m p_k\left(\frac{\varphi_k}{L_k}\right)=\varphi^\intercal \diag{k_1}\varphi+k_2^\intercal \varphi,
\]
where 
\[
	\diag{k_1}=\diag{L}^{-1}\diag{r_1}\diag{L}^{-1},
\]
and	$k_2^\intercal= r_2^\intercal\diag{L}^{-1}$.

Controller gains are selected as follows:
\[
	k_1=\begin{bmatrix}
		0.1623\\
		1.8343
	\end{bmatrix}\times 10^5 \text{ and } k_2=\begin{bmatrix}
		130.7\\
		27.7
	\end{bmatrix}.
\]

Fig.~\ref{fig:powerlosses} depicts cost function $J$ levels (elliptical sections) as well as admissible equilibriums when $Q=Q_r$ for $R=20~\Omega$ and  $R=5~\Omega$ (black dashed lines). 
\begin{figure}
	\centering
	\includegraphics[width=\columnwidth]{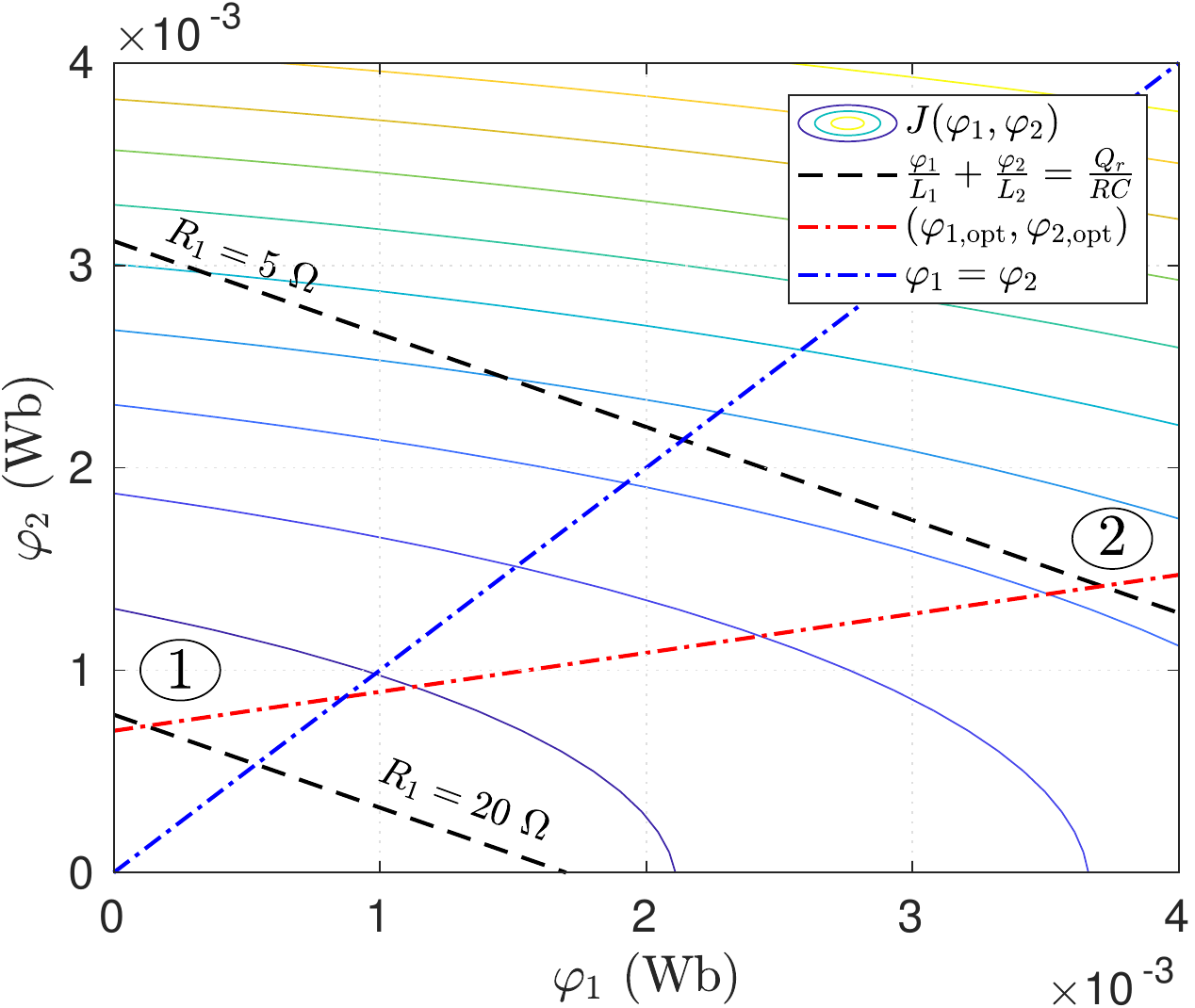}
	\caption{Cost function levels and optimal repartition}
	\vspace{-0.5cm}
	\label{fig:powerlosses}
\end{figure}
The blue dashed line $\varphi_1=\varphi_2$ is the frontier above which more flux goes through the second coil than through the first one. System trajectory evolves below whenever the opposite relationship holds. The optimal locus as a function of the load are located by the red dashed line. On this line power losses are minimal since flux repartition is optimal. Intersection between red dashed line and the black one give the equilibrium point that the closed-loop is expected to reach. Indeed, at the intersection, we ensure that $Q=Q_r$ and $J$ is minimal.

\subsubsection{Experimental results}

Fig.~\ref{fig:resultRinconnue} depicts results of Experiment 2.
\begin{figure}
	\centering
	\includegraphics[width=\columnwidth+0.5cm]{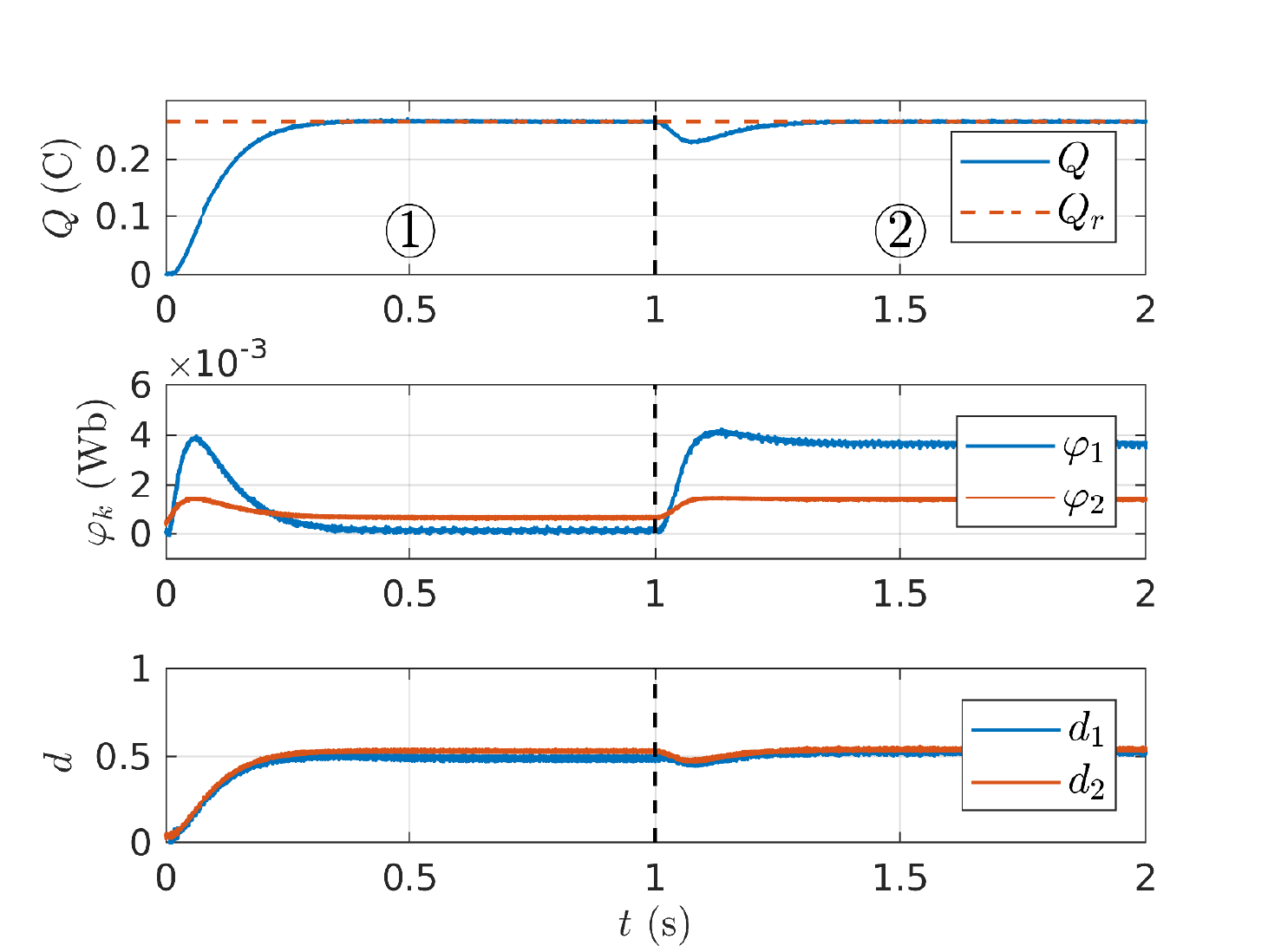}
	\caption{Time results for Exp. 2: energy variables.}
	\vspace{-0.5cm}
	\label{fig:resultRinconnue}
\end{figure}
Subplot 1 depicts charge $Q$ with the reference $Q_r$, Subplot 2 depicts the flux through both inductors while Subplot 3 displays duty cycle of each converter.

Benefit of this last control law is that when the load magnitude changes, the closed-loop converge to the equilibrium point satisfying $Q=Q_r$ and minimizing the cost function $J$. Indeed we can see that for the first phase ($R=20$ $\Omega$), there is more flux passing through the second coil than through the the first one: it corresponds to the \ding{192} of Fig.~\ref{fig:powerlosses}. Yet, for phase 2 ($R=5$ $\Omega$), minimization of power losses gives the priority to the first coil: point \ding{193} of Fig.~\ref{fig:powerlosses} is therefore recovered.

\section{Conclusion}
In this paper, charge dynamics has been separated from flux distribution in the Hamiltonian framework. This separation is related to Casimir functions and translates control objectives. The foremost objective of charge regulation is disconnected from the corresponding secondary objective, corresponding to current repartition and related to Casimir function. Once the separation is done by a change of coordinates, control design can be decomposed in two parts. On the one hand, for the charge regulation, we designed a load independent controller performing an integral action on a non passive output. On the other hand, the control of Casimir functions is designed to minimize the cost function by considering it as the virtual energy of the closed-loop. Experiments have been done in order to highlight the relevance of the proposed control scheme for solving meaningful practical problem of minimization overall losses in the electrical circuit.

Further research will mainly focus an three points. Firstly, we want to integrate other converters (for instance boost converters) in parallel interconnection of converters. This induces non-linearities in the model. Secondly, we want to consider input constraints as duty-cycles are constrained to live in compact set $[0,1]^m$. How to preserve the decoupling between dynamics of flux repartition and output voltage is an open question. Thirdly, robustness with respect to large serial resistance could be interesting. An adaptive way to recover the value of those resistances could be relevant. Finally, dealing with non constant load impedance, like Constant Power Load (CPL), is an practically relevant direction to extend this work.

\bibliographystyle{unsrt} 
\bibliography{bibliothese}

\end{document}